\documentclass[11pt]{article}

\usepackage{fullpage,latexsym,amsthm,amsmath,color,amssymb,url,hyperref}
\usepackage{tikz}
\usepackage{wasysym,marvosym}

\usepackage[T1]{fontenc}
\usepackage[polish,russian,greek,english]{babel}
\usepackage[utf8]{inputenc}
\usepackage{alphabeta}

\newcommand{\mynewtheorem}[2]{
  \newaliascnt{#1}{dummy}
  \newtheorem{#1}[#1]{#2}
  \aliascntresetthe{#1}
  \expandafter\def\csname #1autorefname\endcsname{#2}
}

\newtheorem{lemma}{Lemma}
\newtheorem{claim}{Claim}
\newtheorem{proposition}{Proposition}

\newtheorem{corollary}{Corollary}

\newtheorem{theorem}{Theorem}

\newcommand{\tw}{\mathsf{tw}}

\usepackage{color}
\definecolor{Black}{rgb}{0,0, 0}
\definecolor{mygreen}{rgb}{0, .6, 0}
\definecolor{Blue}{rgb}{0, 0 ,1}
\definecolor{Red}{rgb}{1, 0 ,0}
\definecolor{Other}{rgb}{.1, .6,.7}
\definecolor{Otherother}{rgb}{.9, .3,.4}
\definecolor{Brown}{rgb}{0.5, 0.3, 0.3}
\definecolor{Magenta}{rgb}{0.5, 0, 1}
\definecolor{White}{rgb}{1, 1, 1}
\definecolor{Yellow}{rgb}{.55,.55,0}

\newcommand{\poly}{{\sf poly}}
\newcommand{\bound}[1]{{\bf #1}}
\newcommand{\rep}{{\sf rep}}
\newcommand{\desc}{{\sf desc}}

\date{}
\begin{document}

\title{Data-compression for Parametrized Counting\\ Problems on Sparse graphs}

\author{Eun Jung Kim\thanks{Universit\'{e} Paris-Dauphine, PSL Research University, CNRS/LAMSADE, 75016, Paris, France.}~\thanks{Supported by project ESIGMA (ANR-17-CE40-0028).} \and  Maria Serna\thanks{Computer Science Department \& BGSMath, Universitat Politècnica de Catalunya, Barcelona, Spain.}~\thanks{Partially funded by MINECO and FEDER funds under grants TIN2017-86727-C2-1-R (GRAMM) and MDM-2014-044 (BGSMath), and by AGAUR grant 2017SGR-786 (ALBCOM).} \and Dimitrios M. Thilikos\thanks{AlGCo project-team, LIRMM, CNRS, Université de Montpellier, Montpellier, France.}~\thanks{Department of Mathematics, National and Kapodistrian University of Athens, Greece.}~\thanks{Supported by projects DEMOGRAPH (ANR-16-CE40-0028) and ESIGMA (ANR-17-CE40-0028).}}
\maketitle

\begin{abstract}
\noindent We study the concept of \emph{compactor}, which may be seen as a counting-analogue of kernelization in counting parameterized complexity. 
For a function $F:\Sigma^*\to \Bbb{N}$ and a parameterization $\kappa: \Sigma^*\to \Bbb{N}$, a compactor $({\sf P},{\sf M})$ 
consists of a polynomial-time computable function ${\sf P}$, called \emph{condenser},  and a computable function ${\sf M}$, called  \emph{extractor}, 
such that $F={\sf M}\circ {\sf P}$, and the condensing  ${\sf P}(x)$ of $x$ has length at most $s(\kappa(x))$, for any input $x\in \Sigma^*.$ 
If $s$ is a polynomial function, then the compactor is said to be of polynomial-size.
Although the study on counting-analogue of  kernelization is not unprecedented, 
it has received  little attention so far. 
We study a family of vertex-certified counting problems on graphs that are MSOL-expressible; that is, 
for an MSOL-formula $\phi$ with one free set variable to be interpreted as a vertex subset, we want to 
count all $A\subseteq V(G)$ 
where $|A|=k$ and $(G,A)\models \phi.$ 
In this paper, we prove that every vertex-certified counting problems on graphs that is  \emph{MSOL-expressible} and \emph{treewidth modulable}, when parameterized by $k$, 
admits a polynomial-size compactor on $H$-topological-minor-free graphs
with condensing time 
$O(k^2n^2)$ 
and decoding time  $2^{O(k)}.$ 
This implies the existence of an {\sf FPT}-algorithm  of running time $O(n^2k^2)+2^{O(k)}.$
All aforementioned complexities  are 
under the Uniform Cost Measure (UCM) model where 
numbers can be stored in constant space and arithmetic operations can be done in constant time.
 \end{abstract}

\section{Introduction}

A large part of research on parameterized algorithms has been focused 
on  algorithmic techniques for parametrizations 
of decision problems. However, relatively less effort has been invested 
for solving parameterized counting problems. In this paper, we provide  a general data-reduction concept for counting problems, leading to a formal definition of the  notion of a {\sl compactor}.
Our main result is an algorithmic meta-theorem for the existence of a polynomial size compactor, that is applicable to a 
wide family of problems of graphs. 
%

\subsection{General context}\label{pretends}
\paragraph{Algorithmic meta-theorems.}
Parameterized complexity has been proposed as a multi-variable framework
for coping with the inherent complexity of computational problems. Nowadays, 
it is a mature discipline of modern Theoretical Computer Science and has offered 
a wealth of algorithmic techniques and solutions (see~\cite{CyganFKLMPPS15Para,DowneyF13fund,FlumGrohebook,Niedermeier06inv} for related textbooks).
In some cases, in-depth investigations 
on the common characteristics of parameterized problems
gave rise to algorithmic meta-theorems. Such theorems  typically provide 
conditions, logical and/or combinatorial, for a problem to admit 
a parameterized algorithm~\cite{GroheK11meth,Grohe08logi,Kreutzer08algo,Seese91thes}. Important algorithmic meta-theorems
concern model-checking for Monadic Second Order Logic (MSOL) \cite{Courcelle90them,BoriePT92auto,ArnborgLS91easy,Seese96line} on bounded treewidth graphs and model checking for First Order Logic (FOL) on  certain classes of sparse graphs \cite{FlumG01fixe,FrickG99deci,DawarGK07loca,DvorakKT13firs,DvorakKR10deci,GroheKS2014deci}.

In some cases,
such theorems have a counterpart on {\sl counting}  parameterized problems. Here  the target is to prove 
that counting {\em how many solutions} exist
for a problem is fixed parameter tractable, under some  parameterization of it.
Related meta-algorithmic results   concern counting analogues of 
Courcelle's theorem, proved in~\cite{CourcelleMR00fixe}, stating that counting 
%
%
problems definable in MSOL are fixed-parameter tractable when parameterized by the tree-width of the input graph. Also similar results 
for certain fragments of MSOL  hold when 
parameterized  by the  rank-width of the input graph~\cite{CourcelleMR00fixe}. Moreover, 
it was shown in~\cite{Frick04gene} that 
counting problems definable in first-order logic are fixed-parameter tractable on locally tree-decomposable graphs (e.g. for planar graphs and bounded genus graphs).
%

\paragraph{Kernelization and data-reduction.}
A well-studied concept in parameterized complexity is kernelization. We say that a parameterized problem admits a polynomial kernel 
if there is an algorithm -- called {\em kernelization algorithm} --
that can transform, in polynomial time, every input instance of the problem to an equivalent one,
whose size is bounded by a 
 function of 
the parameter. When this function is polynomial then  we have a {\em polynomial kernel}.
A polynomial kernel permits the drastic 
data-reduction of the problem instances to  equivalent 
``miniatures'' whose size is {\em independent} from the bulk of the input size and is polynomial on the parameter.
That way, a  polynomial kernel, 
provides a preprocessing of  computationally hard problems 
that enables the application of 
exact algorithmic approaches (however still super-polynomial)  
on significantly reduced instances~\cite{LokshtanovMS12kern}. 

\paragraph{Meta-algorithmic results for kernelization.}
Apart from the numerous  advances on the design of polynomial kernels for particular problems, algorithmic 
meta-theorems  appeared also for kernelization. 
The first  result of this type appeared in~\cite{BodlaenderFLPST16meta}, where
it  was proved that certain families of  problems on graphs
admit polynomial kernels on bounded genus graphs. The logic-condition of~\cite{BodlaenderFLPST16meta} is  CMSOL-expressibility  or, additionally, the Finite Integer Index (FII)  property (see~\cite{AbrahamsonF93,BodlaendervA01a,Fluiter97}). Moreover, the meta-algorithmic results of~\cite{BodlaenderFLPST16meta} require additional combinatorial 
properties for the problems in question. The results in~\cite{BodlaenderFLPST16meta}
where  extended in~\cite{FominLST10bidi} (see also~\cite{FominLST16bidi})
where the combinatorial condition for the problem
was related to bidimensionality, while the applicability 
of the results was extended in  minor-closed graph classes. Finally, further extensions appeared in~\cite{KimLPRRSS16line} where, under the bounded treewidth-modulability  property (see~Subsection~\ref{wasteful}),  some of the results in~\cite{FominLST10bidi,BodlaenderFLPST16meta} could 
be applied to more graph classes, in particular those excluding some fixed graph as a topological minor.

\paragraph{Data reduction for counting problems.}
Unfortunately, not much has been done so far in the direction of data-reduction for parameterized counting problems. The most comprehensive work in this direction was done by Marc Thurley~\cite{Thurley07kern}
(see also~\cite{ThurleyDipThes})
who proposed the first formal definition of a kernelization analogue  for parameterized problems called {\sl  counting kernelization}.
In~\cite{Thurley07kern}  Thurley investigated up to which extent 
classic kernelization techniques such as {\sl Buss’ Kernelization}  and {\sl crown decomposition} may lead to counting counterparts of  kernelization. In this direction, he 
provided counting kernelizations  for a series of parameterized counting problems such as and 
$p$-\#{\sc VertexCover}, $p$-{\sc card}-\#{\sc Hitting Set} and $p$-\#{\sc Unique Hitting Set}. 

\paragraph{Compactor enumeration.}
Another framework
for data-reduction on parameterized counting problems
is provided by the notion of a {\em compactor}. In a precursory level,  
it  appeared for the first time in~\cite{DiazST08effi}.
The rough idea in~\cite{DiazST08effi}
was to transform the input of a parameterized counting problem 
to a structure, 
called {\sl the compactor}, whose size is (polynomially) bounded by the parameter
and  such that the enumeration of certain family of objects  (referred as {\em compactor 
enumeration} in~\cite{DiazST08effi}) in the compactor is able to derive the 
number of solutions for the initial instance.  This technique 
was introduced in~\cite{DiazST08effi} for  
counting  restrictive list $H$-colorings and, later in~\cite{NishimuraRT05}, 
for counting generalized coverings and matchings. However none of~\cite{DiazST08effi,NishimuraRT05} provided a general  formal definition of a compactor, while, in our opinion, the work of Thurley provides a legitimate formalization of compactor enumeration.

In this paper, we define formally the concept of a compactor 
for parameterizations of  function problems (that naturally  include counting problems) that is not based on enumeration.
As a first step, we observe that for parameterized function  problems, the existence of a compactor 
is equivalent to the existence of an {\sf FPT}-algorithm, a fact that is also the case for classic kernels 
on decision problems and for counting kernels in~\cite{Thurley07kern}. 
%
%

Under the above formal framework,
we prove an  algorithmic meta-theorem  on the  existence 
of polynomial compactors for a general family of graph problems. In the next subsection,
we define the compactor concept and we present  the related meta-algorithmic results.

%
%

%

\subsection{Our results}
\label{wasteful}

\paragraph{Counting problems and parameterizations.}  
First of all notice that, for a counting problem, it is not possible 
to have a kernelization in the classic sense, that is to produce an 
reduced instance, bounded by a function of $k$,
that is counting-equivalent in the sense that the number of  
solutions in the reduced instance will provide the number  of solutions in the original one.
For this reason we need a more refined notion of data compression
where we transform the input instance to ``structure'', whose 
size is bounded by a function of $k.$ 
 This structure contains enough information (combinatorial and arithmetical)
so as to permit the recovering of the number of the solutions in the initial instance.
We next formalize this idea to the concept of a compactor.

 Let $\Bbb{N}$ be all  non-negative integers and by ${\sf poly}$ the set of all polynomials.
Let $\Sigma$ be a fixed alphabet.
A {\em  parameterized function problem} is a pair $(F,\kappa)$ where $F,\kappa:\Sigma^*\to\Bbb{N}.$
An {\sf FPT}-algorithm for  $(F,\kappa)$ is one that, given  $x\in\Sigma^*$,  
outputs $F(x)$ in $f(\kappa(x))\cdot \poly(|x|)$ steps.
When evaluating the running time, we use the standard Uniform Cost Measure (UCM) model where all 
basic arithmetic computations are carried out in constant time. We  also disregard the size of the numbers
that are produced during the execution of the algorithm.

 \paragraph{Compactors.}
Let $(F,\kappa)$ be a parameterized function problem. A {\em  compactor}
for $(F,\kappa)$ is a pair $(P,M)$  where 
\begin{itemize}
\item $P: \Sigma^*\to \Sigma^*$ is a polynomially computable function, called an {\em condenser}, 
\item $M: \Sigma^*\to \Bbb{N}$ is a computable function, called a {\em extractor}, 
\item $F=M\circ P$, i.e., $\forall x\in\Sigma^*$, $F(x)=(M\circ P)(x)$, and 
\item  there is a recursive function $s:\Bbb{N}\to\Bbb{N}$ where  $\forall x\in \Sigma^{*}\   |P(x)|\leq s(\kappa(x)).$  
\end{itemize}   
We call the function $s$  {\em size} of the compactor $(P,M)$
and, if $s\in{\sf poly}$, we say that $(P,M)$ is a {\em polynomial-size compactor} for $(F,\kappa).$ 
We  call the running time of the algorithm computing  $P$, measured as a function of $|x|$, {\em condensing time} of $(P,M).$ We also call the running time of the algorithm computing $M$,  measured as a function of $\kappa(x)$,  {\em decoding time} of $(P,M).$ We can readily observe the following.

\begin{lemma}
\label{atheists}
A parameterized function problem has an {\sf FPT}-algorithm if and only if 
there is a compactor for it.
\end{lemma}

\begin{proof}[Proof]
Let $(F,\kappa)$ be a parameterized function problem. We assume the UMC model. 
Suppose that an algorithm ${\sf A}$ computes $F(x)$ in time $f(\kappa(x))|x|^{O(1)}$ for any input $x\in \Sigma^*.$ 
Then let ${\sf P}$ be a function defined as 
\begin{equation*}
{\sf P}(x)=
\begin{cases}
F(x) & \text{if }|x| > f(\kappa(x))\\
x		& \text{otherwise. }
\end{cases}
\end{equation*}
Cleary, ${\sf P}(x)$ can be computed in polynomial time since if $|x| > f(\kappa(x))$, then one can compute $F(x)$ by ${\sf A}$ 
in time $f(\kappa(x))|x|^{O(1)}=|x|^{O(1)}.$ Furthermore, ${\sf P}(x)$ has length at most $f(k).$ 
For the extractor ${\sf M}$, we define the image of $z={\sf P}(x)$ under ${\sf M}$ as 
\begin{equation*}
{\sf M}(z)=
\begin{cases}
z & \text{if }|x| > f(\kappa(x))\\
F(z)		& \text{otherwise. }
\end{cases}
\end{equation*}
Note that the function ${\sf M}$ can be computed; in particular $F(z)$ can be computed by ${\sf A}.$ 
Clearly, we have $F={\sf M}\circ {\sf P}$ and $({\sf P}, {\sf M})$ is a compactor for $(F,\kappa).$

Conversely, let $({\sf P},{\sf M})$ be a compactor for $(F,\kappa)$ and a function $s$ be a size of the compactor. 
For any input $x\in \Sigma^*$, we can run an algorithm in time $O(|x|^{O(1)})$ to compute ${\sf P}(x)$ 
and an algorithm in time $g(|{\sf P}(x)|)$ to compute ${\sf M}({\sf P}(x))=F(x).$ 
As $|{\sf P}(x)|\leq s(\kappa(x))$ and the function $g$ can be assumed to be non-decreasing, 
this computes $F(x)$ in time $O(|x|^{O(1)}+g\circ s \circ \kappa(x)).$
\end{proof}

Up to our best knowledge, the notion of compactor as formalized in this paper is new. 
As discussed in Subsection~\ref{pretends}, similar notions have been proposed such as counting kernelization~\cite{Thurley07kern} and 
compactor enumeration~\cite{DiazST08effi}. 
In both counting kernelization and compact enumeration, a mapping from the set of all certificates to certain objects in the new instance is required. 
While this approach comply more with the idea of classic kernelization, it seems to be more restrictive. 
The main difference of our compactor from the previous notions is that 
(the condenser of) a compactor is free of this requirement, which makes the definition more flexible and easier to work with. 
Due to this flexibility and succinctness, we believe that our notion might be  amenable for lower bound machineries akin to those for decision problem kernelizations.

\paragraph{Parameterized counting problems on graphs.}
A {\em structure} is a pair $(G,A)$ where $G$ is a graph and $A\subseteq V(G).$
Given a MSOL-formula $\phi$ on structures and some graph class ${\cal G}$,
we consider the following  parameterized  counting problem
$\Pi_{\phi,{\cal G}}.$
\smallskip

\begin{center}
\noindent\fbox{
\begin{minipage}{13.3cm}
\noindent{$\Pi_{\phi,{\cal G}}$}\\
  {\sl Input}: a graph $G\in \cal G$, an non-negative integer $k.$\\
 {\sl Parameter}: $k.$\\
 {\sl Count}: the number of vertex sets $A\subseteq V(G)$ such that $(G,A)\models \phi$ and $|A|=k.$
\end{minipage}
}
\end{center}
\smallskip

 Formally, $\Pi_{\phi,{\cal G}}$  is  the 
pair $(F_{\phi,{\cal G}},\kappa_{\cal G})$, where 
 $F_{\phi,{\cal G}}: {\cal G}\times\Bbb{N}\to\Bbb{N}$\ is a function
with  $F_{\phi,{\cal G}}(G,k)=|\{A\in{V(G)\choose k}\mid  (G,A)\models \phi\}|$ and $\kappa_{{\cal G}}: {\cal G}\times\Bbb{N}\to\Bbb{N}$\ is the function with $\kappa(G,k)=k.$
To see $\Pi_{\phi,{\cal G}}$ as a counting problem, we consider the relation $R_{\phi,{\cal G}}\subseteq \Sigma^*\times \Sigma^*$ where 
if $(x,y)\in R_{\phi,{\cal G}}$, then $x$ encodes 
$(G,k)\in{\cal G}\times \Bbb{N}$
and $y$ encodes an $A\subseteq V(G)$ such that $|A|=k$ and $(G,A)\models \phi.$ 
Clearly,  $F_{\phi,{\cal G}}(G,k)=|\{y\mid (x,y)\in R_{\phi,{\cal G}}\}|.$

\paragraph{Tree decompositions.} 
A \emph{tree decomposition} of a graph $G$ is a pair $D=(T,\chi)$, where $T$ is a tree
and $\chi: V(T)\rightarrow 2^{V(G)}$
 such that:
\begin{enumerate}
\item $\bigcup_{q \in V(T)} \chi(q) = V(G)$,
\item for every edge $\{u,v\} \in E$, there is a $q \in V(T)$ such that $\{u, v\} \subseteq \chi(q)$, and
\item for each $\{x,y\} \subseteq V(T)$ and each $z\in V(T)$ contained in the unique path of $T$ connecting $x$ and $y$, it holds that  $\chi(x) \cap \chi(y) \subseteq \chi(z).$
\end{enumerate}

\noindent We call the vertices of $T$ {\em nodes} of $D$ and the images of $\chi$ {\em bags} of $D.$ The \emph{width} of a  tree decomposition $D=(T,\chi)$ is $\max\{ |\chi(q)| \mid {q \in V(T)}\} - 1.$ The treewidth of a $G$, denoted by $\tw(G)$, is the minimum width over all tree decompositions of $G.$

\paragraph{Treewidth modulators.}
We say that an instance $(G,k)\in {\cal G}\times \Bbb{N}$ of $\Pi_{\phi,{\cal G}}$ is a {\em null} instance if it has no solutions.
Given a graph $G$, we say that a vertex set $A\subseteq V(G)$ is a {\em $t$-treewidth modulator}  of $G$ if 
the removal of $A$ from $G$ leaves a graph of treewidth at most $t$.
Given an MSOL-formula $\phi$
and a graph class ${\cal G}$, we say that $\Pi_{\phi,{\cal G}}$ is  {\em treewidth modulable} 
if there is a constant $t$ (depending on $\phi$ and ${\cal G}$ only) such that,  for every non-null instance $(G,k)$ of $\Pi_{\phi,{\cal G}}$, $G$ has a $t$-treewidth modulator of size at most $t\cdot k.$

Let ${\cal F}_{H}$ be the class of all graphs that do not contain a subdivision of  $H$ as a subgraph. 
The next theorem states our main result.

\begin{theorem}
\label{hobhouse}
For every graph $H$ and every MSOL-formula $\phi$, if $\Pi_{\phi,{\cal F}_{H}}$ is  treewidth modulable, then there is a   
compactor for $\Pi_{\phi,{\cal F}_{H}}$ of size $O(k^2)$ 
with condensing time 
$O(k^2n^2)$ 
and decoding time  $2^{O(k)}.$
\end{theorem}

As a corollary of the main theorem we have the following.

\begin{corollary}
\label{shutters}
For every graph $H$ and every MSOL-formula $\phi$, if $\Pi_{\phi,{\cal F}_{H}}$ is  treewidth modulable, then  $\Pi_{\phi,{\cal F}_{H}}$ can be solved 
 in   $O(k^2n^2)+2^{O(k)}$ steps.
\end{corollary}

In the above results, the constants hidden in the $O$-notation depend 
on the choice of $\phi$, on the treewidth-modulability constant $t$,  
and on the choice of $H.$

Recall that the above results are stated using the UCM model. 
As for $\Pi_{\phi,{\cal F}_{H}}$, the number of solutions is  $O(n^k)$
and this number can be encoded in $O(k\log n)$ bits. Assuming that summations 
of two $r$-bit numbers can be done in $O(r)$ steps and multiplications of two $r$-bit numbers can be done in $O(r^2)$ steps, then the size of the compactor in Theorem~\ref{hobhouse}
is $O(k^2\log n)$ the condensing and extracting times are $O(k^4n^2\log^2 n)$ and $2^{O(k)}\log^{2}n$ respectively.
Consequently, the running time of 
the algorithm in Corollary~\ref{shutters} is $O(k^4n^2\log^{2}n)+2^{O(k)}\log^{2}n.$

Coming back to the algorithmic meta-theorems on parameterized counting 
problems we should remark
that the problem condition of  Corollary~\ref{shutters} is weaker than MSOL,
as it additionally demands treewidth-modulability. However,
the graph classes where this result applies have unbounded treewidth or rankwidth. 
That way our results can be seen as orthogonal to those of~\cite{CourcelleMR00fixe}.

On the side of FOL, the  problem condition of  Corollary~\ref{shutters} is stronger than FOL, while 
its combinatorial applicability 
includes planar graphs or graphs of bounded genus where, the existing algorithmic meta-theorems 
 require FOL-expressibility (see~\cite{Frick04gene}).

%

\subsection{\bf Outline of the compactor algorithms}

Our approach follows the idea of applying data-reduction based on  protrusion decomposability. This idea was initiated in~\cite{BodlaenderFLPST16meta} for the automated derivation of polynomial kernels on decision problems. The key-concept in~\cite{BodlaenderFLPST16meta} is the notion of a {\sl protrusion}, 
a set of vertices with small neighborhood to the rest of the graph and inducing a graph of small treewidth. 
Also,~\cite{BodlaenderFLPST16meta}  introduced the notion of a {\sl protrusion decomposition}, which is a partition of $G$ 
to $O(k)$ graphs such the first one is a ``center'', of size $O(k)$, and the rest 
are protrusions whose neighborhoods are in the center.

The meta-algorithmic 
machinery of~\cite{BodlaenderFLPST16meta}  is based on the following combinatorial fact:
for the problems in question,
 {\sf YES}-instances -- in our case non-null instances-- admit a protrusion decomposition that, 
 when the input has size $\Omega(k)$,  one of its protrusions is ``big enough''. 
This permits the application of some ``graph surgery'' that consists in 
replacing a big protrusion with a smaller  one and, that way,  creates an 
equivalent instance of the problem (the replacements are based on the MSOL-expressibiliy of the problem). In the case of counting problems,
this protrusion replacement machinery does not work (at least straightforwardly) as  we have 
to keep track, not only of the way some  part of a solution ``invades'' a protrusion, but also 
of the number of {\sl all} those partial solutions. Instead, we take another way that avoids
stepwise protrusion 
replacement. In our approach, the condenser of the compactor  first constructs 
an approximate protrusion decomposition, then,  it computes how many possible partial solutions 
of all possible sizes may exist in each one of the protrusions.
This computation is done by dynamic programming (see Section~\ref{validate}) and produces 
a total set of $O(k^2)$ arithmetic values.  These values, along with 
the combinatorial information of the center of the protrusion decomposition
and the neighborhoods of the protrusions in the center, constitutes the 
output of the condenser. This structure can be stored in $O(k^2)$ space 
(given that arithmetic values can be stored in constant space) and contains 
enough information to obtain the number of all the solutions of the initial instance in $2^{O(k)}$
steps (Section~\ref{validate}).

We stress that the above machinery demands the polynomial-time construction of 
a constant-factor approximation of a  protrusion-decomposition. To our knowledge, this remains an 
open problem in general. So far, no such algorithm has been proposed, even for particular graph classes, mostly because   meta-kernelization machinery in~\cite{BodlaenderFLPST16meta} (and later in~\cite{FominLST16bidi,FominLST10bidi,KimLPRRSS16line,FominLMS12plan})
is based on stepwise protrusion replacement and does not actually need to construct 
such a decomposition. Based on the result in~\cite{KimLPRRSS16line}, we show that 
 that the construction of such an approximate  protrusion decomposition 
is possible on $H$-topological-minor-free graphs, given that 
it is possible to construct an approximate  $t$-treewidth modulator of $G.$
In fact, this can been done in  general graphs using the randomized constant-factor 
approximation algorithm in~\cite{FominLMS12plan}. Responding to the 
need for a deterministic approximation 
we provide a constant-factor
approximation algorithm that finds a $t$-treewidth modulator on $H$-topological-minor free graphs (Section~\ref{renounce}).
This algorithm runs in $O(k^2n^2)$ steps and, besides from being a necessary step of the condenser of our compactor, is of  independent algorithmic interest.

\section{Preliminaries}   \label{excluded}
We use $\Bbb{N}$ to denote the set of all non-negative integers.
Let $\chi:\mathbb{N}^2
\rightarrow \mathbb{N}$ and $\psi: \mathbb{N}
\rightarrow \mathbb{N}.$
We say that $\chi(n,k)=O_{k}(\psi(n))$ if there exists a 
function $\phi:\mathbb{N} \rightarrow \mathbb{N}$
such that  $\chi(n,k)=O(\phi(k)\cdot \psi(n)).$
Given  $a,b\in\Bbb{N}$, we define by $[a,b]=\{a,\ldots,b\}.$ Also, given some $a\in \Bbb{N}$ we define $[a]=\{1,\ldots,a\}.$
Given a set $Z$ and a $k\in \Bbb{N}$, we denote ${Z \choose k}=\{S\subseteq Z\mid |S|=k\}.$

\subsection{Graphs and boundary graphs}

\label{ancients}

\paragraph{Graphs.}
All graphs in this paper are simple and undirected. 
Given a graph $G$, we use $V(G)$ to denote the set of its vertices.  Given a $S\subseteq V(G)$ we denote by $N_{G}(S)$ the set of all neighbours of $S$ 
in $G$ that are not in $S.$ We also set $N_{G}[S]=S\cup N_{G}(S)$ and we use $N(S)$ and $N[S]$ 
as shortcuts of $N_{G}(S)$ and $N_{G}[S]$ (when the index is a graph denoted by $G$).
We define $G-S$ as the graph 
obtained from $G$ if we remove the vertices in $S$, along with the edges incident to them.
The {\em subgraph of $G$ induced  by} $S$ is the graph $G[S]:=G- (V(G)\setminus S).$ Finally, we set $\partial_{G}(S)=N_G(V(G- S)).$
We call $|V(G)|$ the {\em size} of a graph $G$ and $n$ is reserved to denote the size of the input graph for time complexity analysis. 

Given a graph $G$, a {\em subdivision} of $G$ is any graph that is obtained from $G$ after replacing its edges 
by paths with the same endpoints.   We say that a graph $H$ is a {\em  topological minor} of $G$
if $G$ contains as a subgraph some subdivision of $H.$ We also say that $G$ is {\em $H$-topological-minor-free}
if it excludes $H$ as a topological minor. 

\paragraph{Boundaried structures.} \label{speakers}
A {\em labeling} of a graph $G$ is any injective function $\lambda: V(G)\rightarrow\Bbb{N}.$ Given a structure $(G,A),$ we  call $A$ {\em the annotated set} of $(G,A)$ and the vertices in $A$ {\em annotated vertices} of $(G,A).$

A \emph{boundaried structure}, in short a {\em b-structure}, is a triple $\bound{G} = (G,B,A)$ where $G$ is a graph and $B,A\subseteq V(G).$ We say that $B$ is the {\em boundary} of  ${\bf  G}$ 
and $A$ is the {\em annotated set} of ${\bf G}.$ Also we call the vertices of $B$ {\em  boundary vertices} and the vertices in $A$ {\em annotated vertices}.   We use notation ${\cal B}^{(t)}$
to denote all b-structures whose boundary has at most $t$ vertices.
We set $G({\bf G})=G$, $V({\bf G})=V(G)$, $B({\bf G})=B$, $A({\bf G})=A.$ We  refer to $G$ as {\em the underlying graph} of ${\bf G}$
and we always assume that the underlying graph of a  b-structure is accompanied with 
some labelling $\lambda.$ 
Under the presence of such a labelling,
we define the {\em index} of a boundary vertex $v$ as the quantity $|\{u \in B \mid \lambda(u) \leq \lambda(v)\}|$
i.e., the index of $v$   when we arrange the vertices of  $B$ according to $\lambda$ in increasing order.
We extend the notion of index to subsets of $B$ in the natural way, i.e., the index of $S\subseteq B$ consists of the indices of all the vertices in $S.$

A {\em boundaried graph}, in short {\em b-graph}, is any b-structure $\bound{G} = (G,B,A)$ such that $A=V(G).$
For simplicity we use the notation $\bound{G} = (G,B,-)$ to denote b-graphs 
instead of using the heavier notation ${\bf G}=(G,B,V(G)).$ 
For every $t\in\Bbb{N}$, we use $\overline{\cal B}^{(t)}$ to denote the b-graphs in ${\cal B}^{(t)}.$ 
We avoid denoting a boundary graph as an annotated graph as
we want to stress the role of $B$ as a boundary.

We say that two b-structures   
 $\bound{G}_1=(G_1,B_1,A_1)$ and $\bound{G}_2=(G_2,B_2,A_2)$ are {\em compatible}, denoted by $\bound{G}_1\sim \bound{G}_2$,  
if  $A_{1}\cap B_{1}$ and $A_{2}\cap B_{2}$ have the same index and the labeled graphs $G[B_1]$ and $G[B_2]$, where each vertex of $B_i$ is labeled by 
its index,  are identical.
 
%
Given two compatible b-structures ${\bf G}_{1}=(G_{1},B_{1},A_{1})$ and ${\bf G}_{2}=(G_{2},B_{2},A_{2})$, we define 
 ${\bf G}_1\oplus {\bf G}_2$  as the structure $(G,A)$ where 
 \begin{itemize}
 \item the graph $G$ is obtained by taking the disjoint union of $G_1$ and $G_2$ and then identifying boundary vertices of $G_1$ and $G_2$ of the same index, and 
 \item the vertex set $A$ is obtained from $A_{1}$ and $A_{2}$ after identifying 
 equally-indexed   vertices in $A_{1}\cap B_{1}$ and $A_{2}\cap B_{2}.$
 \end{itemize}
 Keep in mind that $(G,A)={\bf G}_{1}\oplus{\bf G}_{2}$ is an annotated graph  and not a b-structure. 
 We always assume that the labels of the boundary of ${\bf G}_{1}$ prevail during the gluing operation, i.e., 
 they are inherited to the identified vertices in $(G,A)$ while the labels of the boundary of ${\bf G}_{2}$ dissapear in $(G,A).$
Especially, when ${\bf G}_1$ and ${\bf G}_2$ are compatible b-graphs, we treat ${\bf G}_{1}\oplus{\bf G}_{2}$ as a graph for notational simplicity.

\paragraph{Treewith of b-structures.}
Given a b-structure  ${\bf G}=(G,B,A)$, we say that the triple $D=(T,\chi,r)$ is a {\em tree decomposition} of ${\bf G}$
if $(T,\chi)$ is a tree decomposition of $G$, $r\in V(T)$, and $\chi(r)=B.$
We see $T$ as a  tree rooted on $r.$ 
The {\em width} of a  tree decomposition $D=(T,\chi,r)$ is the width of the tree decomposition $(T,\chi).$
The treewidth of a b-structure   ${\bf G}$ is the minimum width over all its   tree decompositions
and is denoted by $\tw({\bf G}).$  We use ${\cal T}^{(t)}$ (resp. $\overline{\cal T}^{(t)}$) to denote all b-structures (resp. b-graphs) in ${\cal B}^{(t)}$ (resp. $\overline{\cal B}^{(t)}$) with treewidth at most $t.$

\paragraph{Protrusion decompositions.}
Let $G$ be a graph. Given $\alpha,\beta,\gamma\in\Bbb{N}$, an {\em $(\alpha,\beta,\gamma)$-protrusion decomposition of $G$}
is a sequence of ${\bf G}_{1}=(G_{1},B_{1},-),\ldots,{\bf G}_{s}=(G_{s},B_{s},-)$ of  b-graphs where, given that $X_{i}=V(G_{i})\setminus B_{i}, i\in[s]$, it holds that
\medskip\medskip

\begin{tabular}{lll} 
 &{\bf 1.}\ \mbox{$s\leq \alpha$}& {\bf 2.}\  \mbox{$\forall {i\in[s]},\ {\bf G}_{i}\in\overline{\cal T}^{(\beta)}$}  \\
& {\bf 3.}\  \mbox{$\forall {i\in[s]},\ G_{i}$ is a subgraph of $G$}   & 
{\bf 4.}\ \mbox{$\forall {i,j\in [s]},\ i\neq j\Rightarrow X_{i}\cap X_{j}=\emptyset$}  \\
& {\bf 5.}\  \mbox{$|V(G)\setminus \bigcup_{i\in[s]}X_{i}|\leq \alpha$}  & {\bf 6.}\  \mbox{$\forall {i\in[s]},\ \tw(G[X_{i}])\leq \gamma.$} 
\end{tabular}

\medskip\medskip

\noindent We cal the set $V(G)\setminus \bigcup_{i\in[s]}X_{i}$ {\em center} of the above {$(\alpha,\beta,\gamma)$-protrusion decomposition.}

\noindent {Protrusion decompositions have been introduced in~\cite{BodlaenderFLPST16meta} in the context of kernelization algorithms (see also~\cite{FominLST16bidi,FominLST10bidi}).}
The above definition is a modification of the original one in~\cite{BodlaenderFLPST16meta}, adapted for the needs of our proofs. The only essential modification is the  parameter $\gamma$, used in the last requirement. Intuitively, $\gamma$ bounds the ``internal'' treewidth of each protrusion ${\bf B}_{i}.$

\subsection{MSOL and equivalence on boundaried structures.}
\label{envelope}
%
%

\paragraph{(Counting) Monadic Second Order Logic.} The syntax of Counting Monadic Second Order Logic (CMSO) on graphs includes the logical connectives $\vee,$ $\land,$ $\neg,$ 
$\Leftrightarrow ,$  $\Rightarrow,$ variables for 
vertices, edges, sets of vertices, and sets of edges, the quantifiers $\forall,$ $\exists$ that can be applied 
to these variables and the following  predicates: 
\begin{enumerate}

\item 
$u\in U$ where $u$ is a vertex variable 
and $U$ is a vertex set variable; 
\item 
 $d \in D$ where $d$ is an edge variable and $D$ is an edge 
set variable;
\item 
 $d\multimap u$ where $d$ is an edge variable,  $u$ is a vertex variable, and the interpretation 
is that the edge $d$ is incident with the vertex $u$; 
\item 
 $u \sim v$ where  $u$ and $v$ are 
vertex variables  and the interpretation is that $u$ and $v$ are adjacent; \item 
 equality (``='') of variables representing vertices, edges, sets of vertices, and sets of edges.
 \item the atomic sentence $\mathbf{card}_{q,r}(S)$ that is true if and only if $|S| \equiv q \pmod r.$
\end{enumerate}

If we  restrict the formulas so that  variables are only vertex variables, we obtain the set of FOL-formulas.

\paragraph{Equivalences between  b-structures and b-graphs.}
Let $\phi$ be a MSOL-formula and $t\in\Bbb{N}.$ 
Given two  b-structures   ${\bf G}_{1},{\bf G}_{2}\in {\cal B}^{(t)}$,
we say that ${\bf G}_1\equiv_{\phi,t}{\bf G}_2$ if 
\begin{itemize}
\item  ${\bf G}_{1}\sim {\bf G}_{2}$  and 
\item $\forall {\bf F}\in{\cal B}^{(t)}\ 
{\bf F}\sim {\bf G}_{1}\Rightarrow (
{\bf F}\oplus {\bf G}_1\models\phi\iff {\bf F}\oplus {\bf G}_2\models\phi)$
\end{itemize}
 Notice that $\equiv_{\phi,t}$ is an equivalence relation on ${\cal B}^{(t)}.$ The following result is widely known as Courcelle's theorem
and was proven~\cite{Courcelle90them}. The same result was essentially proven in~\cite{BoriePT92auto} and~\cite{ArnborgLS91easy}.
 The version on structures that we present below appeared in~\cite[Lemma 3.2.]{BodlaenderFLPST16meta}.

 \begin{proposition}
 \label{cockneys}
There exists a computable function $\xi:\Bbb{N}^2\to\Bbb{N}$ such that
for every  CMSO-formula $\phi$ and  every $t\in \Bbb{N}$, the equivalence relation 
$\equiv_{\phi,t}$ has at most $\xi(|\phi|,t)$ equivalence classes.
\end{proposition}

Given a MSOL-formula $\phi$ and under the light of Proposition~\ref{cockneys}, we consider a (finite) set 
${\cal R}_{\phi,t}$ containing one minimum-size 
member from each of the equivalence classes of $\equiv_{\phi,t}.$ 
Keep in mind 
that ${\cal R}_{\phi,t}\subseteq {\cal B}^{(t)}.$ 
Notice that for every ${\bf G}\in {\cal B}^{(t)}$, there is a b-structure in ${\cal R}_{\phi,t}$, we denote it by 
${\sf rep}_{\phi,t}({\bf G})$, such that ${\sf rep}_{\phi,t}({\bf G})\equiv_{\phi,t} {\bf G}.$

\section{Approximating  protrusion decompositions}

\label{renounce}

The main result of this section is  a constant-factor approximation algorithm computing a $t$-treewidth modulator (Lemma~\ref{exorcize}). Based on this we also derive a constant-factor approximation algorithm
for a protrusion decomposition (Theorem~\ref{dishonor}).
%
%
For our proofs we need the following lemma that is a consequence of  the results in~\cite{KimLPRRSS16line}.

\begin{lemma}
\label{workweek}
For every $h$-vertex graph $H$ and every $t\in\Bbb{N}$, there exists a constant $c$  and an algorithm that takes as input an $H$-topological-minor-free graph $G$ and a $t$-treewidth modulator $X\subseteq V(G)$ and  outputs a $(c|X|,c,t)$-protrusion decomposition along with tree decompositions of its b-graphs of width at most $c$, in $O_{h+t}( n )$ steps. 
\end{lemma}

\begin{proof}[Proof]
We may assume that $h\geq 3$ and $G$ is an $H$ topological-minor-free graph.  It is known that for any $h\geq 3$, there is a $\beta_{h}>h$ 
such that every $K_h$-topological-minor-free graph on $n$ vertices has at most 
$\beta_{h} n$ cliques of any non-negative size (including 0 or 1 (see e.g.~\cite{LeeO15numb,FominOT10rank,ReedW09alin,NorineSTW06prop}). 
We first prove the following claim:

\noindent{\em Claim 1}: If $X\subseteq V(G)$ and $C_1,\ldots , C_p$  is a collection of pairwise vertex-disjoint connected subsets  of $V(G)\setminus X$ such that $|N(C_i)\cap X|\geq h$, then  $p\leq \beta_{h} |X|.$

\noindent{\em Proof of claim}:
Let $W_0=G[X].$ 
Consider the following procedure iterating over $i=1,\ldots , p$: for each $i$, we choose two vertices $u,v \in N(C_i)\cap X$ which are non-adjacent in $W_{i-1}$ and 
add the edge $uv$ in $W_{i-1}.$ Let $W_i$ be the resulting graph. For each $i\geq 0$, observe that $W_i$ is a topological minor of $G$, therefore $W_i$ itself is $H$-topological-minor-free. 
Since an $H$-topological-minor-free graph does not contain as a subgraph a clique on $h$ vertices and $|N(C_i)\cap X|\geq h$, not every vertex pair in $N(C_i)\cap X$ is adjacent in $W_{i-1}.$ 
This means that the procedure will be carried out up to  the $p$-th iteration. We conclude that 
$W_p$ has at least $p$ edges, therefore also at least $p$ cliques.  As $W_p$ is $H$-topological-minor-free, we derive the claim. 

 In~\cite[Algorithm 1]{KimLPRRSS16line}, an $O_{r+t}(n)$-time algorithm ${\sf A}$,  is presented which takes as input a graph $G$, a $t$-treewidth modulator $X\subseteq V(G)$, and a positive integer $r$, and outputs some $Y_0\supseteq X$ and a collection ${\cal C}$ of pairwise vertex-disjoint connected subsets of $V(G)-X$ such that
\begin{enumerate}
\item[(i)] $|Y_0|\leq |X|+2t|{\cal C}|$ ,
\item[(ii)] for every $C\in {\cal C}$,\  $|N(C)\cap X|\geq r$, and 
\item[(iii)] for every connected component $Z$ of $G\setminus Y_0$, $|N(Z)\cap X|<r$ and $|N(Z)\cap Y_0|<r+2t.$ 
\end{enumerate}
The set $Y_{0}$ and the collection ${\cal C}$
is produced by~\cite[Algorithm 1]{KimLPRRSS16line}. In particular the sets in ${\cal C}$ are 
the sets denoted by  ``$C_{B}$'', while condition (i) is  justified by the course of that algorithm.
Condition (iii) is proved in~\cite[Lemma 7]{KimLPRRSS16line}.
Apply the algorithm ${\sf A}$ for $r:=h$ in time $O_{h+t}(n)$ and observe that, from Claim~1, $|{\cal C}|\leq \beta_{h}|X|$, therefore 
\begin{eqnarray}
|Y_{0}|\leq (1+2t\beta_{h})|X|\label{survival}
\end{eqnarray}
Let 
${\cal Y}_1,\ldots , {\cal Y}_{s}$ be the partition of the connected components of $G-Y_0$ into maximal collections of connected components of $G-Y_0$  that
have the same neighborhood in $Y_0.$ 
We set $Y_{i}=\bigcup_{Y\in {\cal Y}_i}Y$ and we call $Y_{1},\ldots,Y_{s}$ {\em clusters} of $G-Y_0.$ 
Clearly, the clusters of $G-Y_0$ can be found in linear time. 

\noindent{\em Claim 2.}  the number $s$ of clusters of  $G-Y_0$ is at most $\beta_{h}|Y_0|.$ 

\noindent{\em Proof of Claim.} 
Let $I\subseteq \{1,\ldots, s\}$ be a maximum set of indices such that
there exist $|I|$ pairwise-distinct vertex pairs $(u_i,v_i)$ for $i\in I$ satisfying $u_i,v_i\in N(V(Y_i))$ and $u_i\neq v_i.$ 
Among all such sets of indices, we select $I$ so as to minimize $\sum_{i\in  I} |N(V(Y_i))|$, and fix a pair $(u_i,v_i)$ for each $i\in I.$
Consider the graph $K=(Y_0,\{u_iv_i: i\in I\})$ and keep in mind that $K$ is a topological minor of $G$, thus it is $H$-topological-minor-free. 

By maximality of $I$, each $N(V(Y_i))$ for $i\in [s]\setminus I$ is a clique  in $K$ (whose size is some non-negative number). 
Moreover, it holds that $|N(V(Y_i))|\neq 2$ for all $i\in [s]\setminus I.$ Indeed, if $N(V(Y_i))=\{u,v\}$, then the maximality of $I$ 
implies that an edge between the vertex pair $(u,v)$ has been added to $K$ for some $j\in I.$ Because $Y_i$ and $Y_j$ are distinct clusters, it follows that $|N(V(Y_i))|< |N(V(Y_j))|.$ 
However, $I\cup \{i\} \setminus \{j\}$ provides the same graph $K$, while the sum of neighborhood sizes over the index set strictly decreases, a contradiction.

Now, set $\phi(i)=\{u_i,v_i\}$ for every $i\in I$, and $\phi(i)=N_G(Y_i)$ for every $i\in [s]\setminus I.$ 
By the previous argument, it is easy to see  that $\phi$ is an  injective mapping $\phi$ from $[s]$ to vertex sets of cliques in $K.$ 
As $K$ is $H$-topological-minor-free,  $s\leq \beta_h|Y_0|$ and the claim holds.

Lastly, we set ${\bf G}_i=(G[N[Y_i]],N(V(Y_i)),-), i\in[s]$ and argue that ${\bf G}_1,\ldots , {\bf G}_{s}$ is a $(c|X|,c,t)$-protrusion decomposition, where $c:=\beta_h(1+2t\beta_{h}).$ Consider a tree decomposition $(T_i,\chi_i)$ of $G[Y_i]$ of width at most $t$, which exists 
since each connected component of $G-Y_0$ is a subgraph of $G-X$ (because $Y_0\supseteq X$) and $X$ is a $t$-treewidth modulator. Each $(T_{i},\chi_{i})$ can be 
computed in $O_{t}(n)$ steps, using the algorithm in~\cite{Bodlaender96ali}.  
By adding a root node $r_i$ to an arbitrary node of $T_i$ with $\chi'(r_i)=N(V(Y_i))$ and letting $\chi'(x)=\chi(x)\cup N(V(Y_i))$ for every node $x\in V(T_{i})$, 
we can obtain a tree decomposition $(T'_i,\chi',r_i)$ of the b-graph ${\bf G}_i.$ Observe that the width of $(T'_i,\chi',r_i)$ is  at most $3t+h\leq h(t+1)\leq c$, by the third condition in the output of ${\sf A}.$
Now, it is straightforward to verify that
\begin{enumerate}
\item $s \leq c|X|$, because of \eqref{survival} and  Claim 2,
\item $\forall {i\in[s]},\ {\bf G}_{i}\in\overline{\cal T}^{(c)}$,
\item  $\forall {i\in[s]},\ G[N[Y_i]]$ is a subgraph of $G$,
\item $\forall {i,j\in [s]},\ i\neq j\Rightarrow Y_{i}\cap Y_{j}=\emptyset$,
\item $|V(G)\setminus \bigcup_{i\in[s]}Y_{i}| =|Y_0|\leq (1+2t\beta_{h})|X|\leq c|X|$, because of  \eqref{survival}, and 
\item $\forall {i\in[s]},\ \tw(G[Y_{i}])\leq t.$
\end{enumerate}
The last item holds because for each $i\in[s]$, $G[Y_{i}]$ is a subgraph of $G\setminus X$ and $X$ is a $t$-modulator of $G.$
\end{proof}


As a consequence of Lemma~\ref{workweek}, as long as the input graph $G$ has many  
vertices (linear in $k$), there is a vertex set $Y$ whose (internal) treewidth is at most $t$ and contains sufficiently many vertices. 
The key step of the approximation algorithm, to be shown in the next lemma, is to replace $N[Y]$
with a smaller graph of the same `type'. 
Two conditions are to be met during the replacement: first, the minimum-size of a $t$-treewidth modulator remains the same. Secondly, 
a $t$-treewidth modulator of the new graph can be `lifted' to a $t$-treewidth modulator of the graph before the replacement without increasing the size. 

\begin{lemma}
\label{exorcize}
For every $h$-vertex graph $H$ and every $t$, there is a constant $c$, depending on $h$ and $t$, and an algorithm 
that, given a graph $G\in{\cal F}_{H}$ and $k\in\Bbb{N}$, 
either outputs an $t$-treewidth-modulator of $G$ of size at most $c\cdot k$ or reports that no $t$-treewidth 
modulator of $G$ exists with size at most $k.$ This algorithm runs in $O_{h+t}(n^2)$ steps.
\end{lemma}

\begin{proof}
Let $c'$ be the constant from Lemma~\ref{workweek}; for any $t$-treewidth modulator $X$ of $G$, there is 
 $(c'|X|,c',t)$-protrusion decomposition. 
We set $c=c'(b+1)$, and the constant $b$ shall be fixed later. 
We first observe that there is an MSOL-formula $\phi_{t}$ such that given a structure $(G,A)$, 
$A$ is a $t$-treewidth modulator iff $(G,A)\models \phi_{t}.$  To see this, take into account 
that for every $t$, there exits a set ${\cal O}_{t}$ of graphs such  that $A$ is a $t$-treewidth modulator of $G$ iff 
$G\setminus A$ does not contain any subdivision of a graph in ${\cal O}_{t}.$
As topological minor containment can be expressed in CMSOL, one can 
use the this observation to construct $\phi_{t}$, as required. 

The next claim states that whenever $G$ has sufficiently many vertices, either it contains a boundaried graph on a vertex set $Y$ (which we shall replace by a boundaried graph of the same `type' with strictly smaller size) 
or it does not have a $t$-treewidth modulator of size at most $k.$

\begin{claim}\label{trustful}
If $G$ has a $t$-treewidth modulator of size at most $k$
and $|V(G)|> ck$, then 
$G$ contains a vertex subset $Y$ such that $\partial_{G}(Y)\leq 2c'+1$, $b<|Y|\leq 2b$ and $\tw(G[Y\setminus \partial_{G}(Y)])\leq t.$
\end{claim}
\noindent{\em Proof of claim:} 
By the assumption and the choice of $c'$, we know that $G$ has a 
 $(c'k,c',t)$-protrusion decomposition ${\bf H}_{1},\ldots,{\bf H}_{s}.$ 
As $|G|>c'k+c'k\cdot b$,  there will be an $h\in [s]$ such that $|{\bf H}_h|>b.$
We set ${\bf H}_{h}=(H_{h},B_{h},-)$ and $X_{h}=V(H_{h})\setminus B_{h}.$
Notice that $Z=V({H}_{h})$ satisfies 
$|\partial_{G}(Z)|\leq c'$ and $\tw(G[Z\setminus \partial_{G}(Z)])\leq t.$
From Lemma~\cite[]{} $V(H_h)$ contains a subset $Y$ 
where  $|\partial_{G}(Y)|\leq 2c'+1$ and $b<|Y|\leq 2b.$ The same proof implies that  $Y\setminus \partial_{G}(Y)\subseteq X_{h}.$ This,  
together with the fact that $\tw(G[X_{h}])\leq t$, imply that $\tw(G[Y\setminus \partial_{G}(Y)])\leq t.$\qed

Now we want to fix the constant $b.$ Let $d=2c'+1.$
Consider a subset $\overline{\cal T}^{(d)}_t\subseteq \overline{\cal T}^{(d)}$ which 
consists of b-graphs ${\bf G}=(G,B,-)$ of treewidth at most $d$ satisfying $\tw(G\setminus B)\leq t.$ 
Clearly, for a vertex set $Y$ of a graph $G$ satisfying the conditions of Claim~\ref{trustful}, 
the b-graph $(G[Y],\partial_G(Y),-)$ is a member of $\overline{\cal T}^{(d)}_t.$
For ${\bf G}\in \overline{\cal T}^{(d)}_t$ and an integer $i\in\Bbb{N}$, 
we define $${\sf type}_{i}({\bf G})=\{\rep_{\phi_t,d}(G,B,A)\mid  A\in {V(G)\choose i}\}.$$ 
and write ${\sf type}_{\leq \ell}({\bf G})=\langle{\sf type}_{0}({\bf G}),\ldots,{\sf type}_{\ell}({\bf G})\rangle.$
We say that ${\bf G}_1$ and ${\bf G}_2$ are \emph{$\ell$-type-equivalent} for $t$ if ${\sf type}_{\leq \ell}({\bf G}_1)={\sf type}_{\leq \ell}({\bf G}_2).$
Intuitively, that two b-graphs ${\bf G}_1$ and ${\bf G}_2$ have the same $\ell$-type 
for $t$ means that for any (partial) $t$-treewidth modulator of ${\bf G}_1$, ${\bf G}_2$ has a (partial) $t$-treewidth modulator of the same size and achieving 
an identical `state' on the boundary, and vice versa. 
Notice that when ${\bf G}_1$ and ${\bf G}_2$ are $\ell$-type-equivalent for some $t$, they have the same boundary size.

\begin{claim}\label{glossing}
Let ${\bf G}_1,{\bf G}_2 \in \overline{\cal T}^{(d)}_t$ be $d$-type-equivalent.
Let ${\bf F}\in \overline{\cal T}^{(d)}$ be an arbitrary b-graph with $|B({\bf F})|=|B({\bf G}_1)|.$ 
Then for every $k\in\Bbb{N}$, ${\bf F}\oplus{\bf G}_{1}$ has a $t$-treewidth modulator of size at most $k$  
if and only if  ${\bf F}\oplus{\bf G}_{2}$ does. 
\end{claim}
\noindent {\em Proof of claim:}
We remind that ${\bf G}_1$ and ${\bf G}_2$ are compatible because of ${\sf type}_0({\bf G}_1)={\sf type}_0({\bf G}_2)$, and 
thus ${\bf F}\oplus{\bf G}_{2}$ is well-defined. Let ${\bf G}_i=(G_i,B_1,-)$ for $i=1,2$ and let ${\bf G}=(G,B,-).$
Suppose that $A$ is  a $t$-treewidth modulator of ${\bf F}\oplus {\bf G}_1$ of size at most $k.$ 
Furthermore, we can assume that $|A\cap V(G_1)|\leq d.$ Indeed, if this is not the case then, because 
$\tw(G_1\setminus B_1)\leq t$,  we could replace 
$A$ by $A^*=(A\setminus  V(G_1))\cup B_1$ that is also a $t$-treewidth modulator of $G$ and $|A^*|\leq |A|.$

We set $L_1=V(G_{1})\cap A$, $L=A\setminus (V(G_{1})\setminus B_1)$, and note that $|L_1|\leq d.$ 
As  ${\sf sign}_{\leq d}({\bf G}_{1})={\sf sign}_{\leq d}({\bf G}_{2})$, 
we also have that  ${\sf sign}_{|L_1|}({\bf G}_{1})={\sf sign}_{|L_1|}({\bf G}_{2}).$ 
This implies that there is an $L_{2}\subseteq{V(G_2)\choose |L_1|}$ such that $(G_{1},B_1,L_1)\equiv_{\phi_t,d} (G_{2},B_2,L_{2}).$
From this equivalence, we derive that $(F,B,L)\oplus (G_{1},B_1,L_{1})\models\phi_t \iff (F,B,L)\oplus (G_{2},B_2,L_{2}) \models \phi_t$
or, equivalently, $(G,A)\models\phi_t\iff (G',A') \models \phi_t.$ 
Finally observe that $|L\cup L_{1}|=|L\cup L_2|.$ 
The opposite direction can be proved in the same way (we only exchange the roles of ${\bf G}_1$ and ${\bf G}_2$).\qed

Since ${\sf type}_{i}({\bf G})$ is a subset of the representatives ${\cal R}_{\phi_t,d}$ and $|{\cal R}_{\phi_t,d}|\leq \xi(|\phi_t|,d)$ for some function $\xi$ by Proposition~\ref{cockneys}, 
there are at most $2^{\xi(|\phi_t|,d)\cdot (\ell+1)}$ distinct $\ell$-types for $t.$ 
Since $d$-type-equivalence is an equivalence relation on $\overline{\cal T}^{(d)}_t$, 
we can partition the set $\overline{\cal T}^{(d)}_t$ of b-graphs into $2^{\xi(|\phi_t|,d)\cdot (d+1)}$-many equivalence classes under $d$-type-equivalence. 
Let ${\cal R}$ be the set containing a minimum-size b-graph of each equivalence class of $\overline{\cal T}^{(d)}_t.$ 
We set the constant $b:=\max_{{\bf G} \in {\cal R}}|{\bf G}|. $

Consider a routine {\sf B} on $G$ which outputs a vertex subset  $Y\subseteq V(G)$ such that 
$|\partial_{G}(Y)|\leq d$,  $\tw(G[Y\setminus \partial_{G}(Y)])\leq t$ and $b<|Y|\leq 2b$, or reports that no such $Y$ exists.
Notice that there exist a FOL-formula $\psi_{b,t}$ 
such that $G\models \psi_{b,t}$ if and only if a desired vertex subset $Y$ exists.  
As proved in~\cite{DvorakKR10deci}, model-checking for FOL-formulas can be done in linear time for classes of graphs with bounded expansion (which include $H$-topological-minor-free graphs). Moreover, according to~\cite{KazanaS13enum}, in the same graph classes, answers to first-order queries can be enumerated with constant delay after a linear time preprocessing. 
Therefore, there exists a routine {\sf B} that, given a graph $G$, either correctly reports that $G$ does not contain  a set $Y$ as above or outputs one in $O_{h,t}(n)$-steps.

We present the approximation algorithm. 
Starting from $G_1:=G$ and iterating over $i=1,\ldots $, we run the routine {\sf B} as long as $|G_i|> ck.$ 
If the routine ${\sf B}$ reports that no such set $Y$ exists, then the algorithm reports that $G$ contains no $t$-treewidth modulator of size at most $k$ exists and terminates. 
Otherwise, we set $G_{i+1}:={\bf F}\oplus {\bf G}'$ where $G_i={\bf F}\oplus (G_i[Y],\partial_{G_i}(Y),-)$, and ${\bf G}'$ is the member of ${\cal R}$ which is 
$d$-type-equivalent with $(G_i[Y],\partial_{G_i}(Y),-).$ 
Clearly, each iteration can be performed in $O_{h+t}(n)$ steps, which is the runtime of the routine {\sf B}. 
At each iteration $i$, we have $|{\bf G}'|\leq b<|Y|$ and thus the algorithm terminates in at most $n$ iterations.
Therefore, in $O_{h+t}(n^2)$ steps, we  either report that $G$ has no $t$-treewidth modulator of size $\leq k$, or 
produce a sequence $G=G_{1},G_{2},\ldots,G_{q}$ of graphs with $|G_q|\leq ck.$ 

Let us see the correctness of the algorithm.
If the answer of {\sf B} is negative at iteration $i$, the condition $|V(G_i)|>ck$ and Claim~\ref{trustful} 
implies that $G_i$ does not contain any $t$-treewidth modulator of size at most $k.$ Since 
the minimum-size of a $t$-treewidth modulator remains the same for $G$ and $G_i$ by Claim~\ref{glossing}, 
$G$ does not contain any $t$-treewidth modulator of size $k.$ 

Suppose the algorithm produces a sequence $(G=)G_{1},G_{2},\ldots,G_{q}.$ 
Notice that $A_q=V(G_q)$ is a $t$-treewidth modulator for $G_q$ having at most $ck$ vertices. 
By `lifting' this solution iteratively, we can produce a sequence $A_{q-1},\ldots,A_{2},A_{1}=A$
of  $t$-treewidth modulator for the graphs $G_{q-1},\ldots,G_{1}=G$, each of size at most $ck.$ 
Formally, given a $t$-treewidth modulator $A'$ of $V(G_{i+1})$ obtained by replacing $(G_i[Y],\partial_{G_i}(Y),-)$ by its representative $(G',B',-)$ in ${\cal R}$, 
a $t$-treewidth modulator of $G_i$ can be constructed by taking  
$A'\setminus  V(G_{i+1})\cup L_i$, where $L_i$ is a vertex subset of $Y$ such that $(G_i[Y],\partial_{G_i}(Y),L_i)\equiv_{\phi_t,d} (G',B',A'\cap V(G')).$ 
Since $(G',B',-)$ is a b-graph with $|B'|\leq d$ and $\tw(G'\setminus B')\leq t$, we may assume that $|A'\cap V(G')|\leq d.$  
Accordingly, the $d$-type-equivalence between $(G_i[Y],\partial_{G_i}(Y),-)$ and $(G',B',-)$ ensures the existence of such $L_i.$ 
The actual set $L_i$ can be computed in $O(1)$ steps as the size of $Y$ is bounded by the constant $2b$ and the equivalence $\equiv_{\phi_t,d}$ 
can be tested in $O(1)$ steps as well on the two constant-sized b-structures $(G_i[Y],\partial_{G_i}(Y),L_i)$ and $(G',B',A'\cap V(G')).$ 
Observe that $A'\setminus  V(G_{i+1})\cup L_i$ is a  $t$-treewidth modulator of $G_i$ of size at most $ck$.
\end{proof}

Notice that the above lemma, with worst running time, is also a consequence of the recent results in~\cite{GuptaLLMW18losi}.
We insist to the above statement of Lemma~\ref{exorcize}, as we are interested for a quadratic time approximation
algorithm for protrusion decompositions. Indeed, based on Lemma~\ref{exorcize} we can prove the following that is the main result of this section.

\begin{theorem}
\label{dishonor}
Let $H$ be an $h$-vertex graph and $\phi$ be a MSOL-formula that is treewidth modulable.
Then there is a constant $c$, depending on $h$ and $|\phi|$,  and an algorithm that, 
given an input  $(G,k)$ of $\Pi_{\phi,{\cal F}_{H}}$, either reports no $A\subseteq V(G)$ with $(G,A)\models \phi$ has size at most $k$ or 
outputs a $(ck,c,c)$-protrusion decomposition of $G$ along with tree decompositions of its b-graphs, each of width at most $c.$ This algorithm runs in $O_{|\phi|+h}(n^2)$ steps.
\end{theorem}

\begin{proof}[Proof]
Recall that, as $\phi$ is  treewidth modulable, there is a constant $c'$, depending on $|\phi|$, such that 
if there is a set $A$ of size $k$ with $(G,k) \models \phi$, then $G$ has a ${c'}$-treewidth modulator $S$ of size $c'\cdot k.$

The algorithm calls the algorithm of Lemma~\ref{exorcize} for $t=c'$ and 
for $k':=c'\cdot k$ instead of $k.$ Let $c''$ be the constant of Lemma~\ref{exorcize} (depending on $h$ and $c'$ that, in turn, depends on  $|\phi|$).
If it reports that no $c'$-treewidth 
modulator of $G$ exists with size $k'$, then it safely reports that no set $A$ of size $k$ satisfies $(G,A)\models \phi.$
Suppose now that the algorithm of Lemma~\ref{exorcize} returns a  
$c'$-treewidth-modulator of $G$ of size at most $k'':=c'' \cdot k'.$ Then, according to Lemma~\ref{workweek} there is a constant $c'''$, depending on $h$ and $c'$,  
and an algorithm that outputs a $(c'''k'',c''',c')$-protrusion decomposition of $G$ along with the tree decompositions of its b-graphs.
As the overall running time of the algorithm is dominated by the one of Lemma~\ref{exorcize}, the theorem follows if we set $c=c'\cdot c''\cdot c'''.$
\end{proof}

\section{The compactor}
\label{validate}

By Theorem~\ref{dishonor}, we may assume that a $(tk,t,t)$-protrusion decomposition ${\bf G}_{1},\ldots, {\bf G}_{s}$ of $G$, with ${\bf G}_i=(G_i,B_i,-)$, is given for some  $t$. 
For counting the sets $A\subseteq V(G)$ of size at most $k$ with $(G,A)\models \phi$, we view such a set $A$ 
as a union of $A_0\cup A_1\cup \cdots A_s$, where $A_0$ is the subset of $A$ residing in the the center of the  decomposition, 
and $A_i=A\cap V({\bf G}_i)$ for each $i\in [s]$. 
Suppose that $A'_i\subseteq V({\bf G}_i)$ for some $i\in [s]$ satisfies $(G_i,B_i,A_i)\equiv_{\phi,t}(G_i,B_i,A'_i)$ and $|A_i|=|A'_i|$. 
Then, $(A\setminus A_i) \cup A'_i$ has the same size as $|A|$ and we have $(G,A\setminus A_i \cup A'_i)\models \phi$. 
In other words, $A'_i$ and $A_i$ are indistinguishable when seen from outside of ${\bf G}_i$. 

The basic idea of the condenser is to replace all the occurrences of such sets $A'_i$ (include $A_i$ itself) with  $O(1)$-bit information; 
that is, the number of such sets, the size of $|A'_i|$, and the equivalence class containing $(G_i,B_i,A'_i)$.
%
Formally, for the given CMSO-formula $\phi$ and $t\in \Bbb{N}$, 
 we define the function  ${\sf \#sol}_{\phi,t}$ so that for each ${\bf R}\in{\cal R}_{\phi,t}$, ${\bf G}:=(G,B,-)\in\overline{\cal T}^{(t)}$,  we set
 
\begin{eqnarray*}
{\sf \#sol}_{\phi,t}({\bf R},{\bf G},k) & =& |\{A\in {V(G)\choose k}\mid  {\bf R}\equiv_{\phi,t}(G,B,A)\}|.
\end{eqnarray*}

\noindent This function can be fully computed in linear time on a b-graph of bounded treewidth.

\begin{lemma}
\label{supports}
For every CMSO-formula $\phi$ and every $t\in\Bbb{N}$, there exists an algorithm that, given a ${\bf G}\in\overline{\cal T}^{(t)}$
and a tree decomposition of ${\bf G}$ of width at most $t$, outputs ${\sf \#sol}_{\phi,t}({\bf R},{\bf G},k')$ for every $({\bf R},k')\in{\cal R}_{\phi,t}\times[0,k].$ This computation takes  $O_{|\phi|,t}(nk^2)$ steps.
\end{lemma}

The proof of Lemma~\ref{supports} is based on a  dynamic programming procedure. 
This may follow implicitly from the proofs of Courcelle's theorem (see~\cite{CourcelleM93mona,CourcelleMR00fixe}). However, we could not find explicit 
statement of it, we present it  for completeness.

\begin{proof}[Proof]
We may assume  that the  tree decomposition $D=(T,\chi,r)$ of ${\bf G}=(G,B,-)$ has the following properties.
\begin{itemize}
\item If $x\in V(T)$ has two children $x_{1}$, $x_{2}$, then $\chi(x)=\chi(x_1)=\chi(x_2).$
\item if $x$ has one child $y$, then the symmetric difference of $\chi(x)$ and $\chi(y)$ contains exactly
one vertex.
\item there is no vertex in $T$ with more than 2 vertices.
\end{itemize}
The above is a so-called {\sl nice tree decomposition} where the boundary $B$ is the root node (see~\cite{BodlaenderK96effi}).

Given a $q\in V(T)$ we denote its set of descendants in $T$, rooted on $r$, including $q$, by $\desc_{T}(q).$ 
For each $q\in V(T)$,
we set $T_{q}=T[\desc_{T}(q)]$ and  we denote by ${\bf G}_{q}$ the b-structure   $(G_{q},B_{q},-)$
where 
$$G_{q}=G[\bigcup_{q'\in V(T_{q})}\chi(q')] \mbox{~and~}B_{q}=\chi(q), 
$$
Notice that if $b\in\desc_{T}(a)$  then $G_{b}$ is a subgraph of $G_{a}.$
Finally,  for every $q\in V(T)$, we  set $\chi_{q}=\chi|_{V_{q}}$, and observe that $D_{q}=(T_{q},\chi_{q},q)$  is a tree decomposition of the b-structure  ${\bf G}_{q}.$

For each node $x$ of $T$, the algorithm will compute the value ${\sf \#sol}_{\phi,t}({\bf R},{\bf G}_{x},k')$  
for every $({\bf R},k') \in {\cal R}_{\phi,t}\times[0,k]$, provided the corresponding values for the children of $x$. 
We distinguish the following cases:

%

\noindent{{\em Case 1:} $x$ is a leaf.} We assume the values of ${\sf \#sol}_{\phi,t}$ are initially set to zero. 
For every $A\subseteq B_x$, we identify the unique $\rep(G_x,B_x,A)$ and 
increase the value of ${\sf \#sol}_{\phi,t}(\rep(G_x,B_x,A),{\bf G}_{x},|A|)$ by 1. Clearly, this correctly computes the function ${\sf \#sol}_{\phi,t}$ at $x.$ 
The b-structure $\rep(G_x,B_x,A)$ can be found by testing the equivalence of $(G_x,B_x,A)$ and each member of ${\bf R}\in {\cal R}_{\phi,t}$ compatible with $(G_x,B_x,A).$ 
Note that the equivalence testing can be done by considering every member ${\bf F}\in {\cal R}_{\phi,t}$ compatible with  ${\bf R}$ and 
see if ${\bf F}\oplus {\bf R}\equiv_{\phi,t} {\bf F}\oplus (G_x,B_x,A).$ 
Since the number of elements in ${\cal R}_{\phi,t}$ is at most $\xi(|\phi|,t)$ by Proposition~\ref{cockneys} and the size of $G_x$ is at most $t$, 
$\rep(G_x,B_x,A)$ can be identified in $O_{t,h}(1)$ steps.

\noindent{{\em Case 2:} $x$ has one child $y$ and  $\{v\}=\chi(x)\setminus \chi(y)$.}
Given a ${\bf R}\in{\cal R}_{\phi,t}$ where  the boundary size of ${\bf R}$ is 
the same as the boundary size of ${\bf G}_{x}$,  
 we  set 
${\bf H}_{\bf R}=(G_x[B_x],B_{y},A_{\bf R})$, where $A_{\bf R}$ is 
the annotated boundary vertices of ${\bf G}_{x}$ that have the same indices as the annotated boundary vertices of ${\bf R}$.
For every such ${\bf R}\in{\cal R}_{\phi,t}$, we define 
$$\frak{P}_{x}({\bf R})=\{{\bf R}'\in {\cal R}_{\phi,t} \mid {\bf R}'\sim {\bf H}_{\bf R} \mbox{~and~}(G({\bf R}'\oplus {\bf H}_{\bf R}),B_{x},A({\bf R}'\oplus {\bf H}_{\bf R}))\equiv_{\phi,r}{\bf R}\}.$$
We also define $b_{\bf R}$ to be 1 or 0 depending on whether 
the vertex in the boundary of ${\bf R}$, that has the same index as $v$, is an annotated vertex of ${\bf R}$ or not.
Observe that
$$\#{\sf sol}_{{\phi},t}({\bf R},{\bf G}_{x},k')=\! \!\sum_{{\bf R}'\in\frak{P}_{x}({\bf R})}\!\!\#{\sf sol}_{{\phi},t}({\bf R}',{\bf G}_{y},k'-b_{\bf R}),$$
therefore we can compute the values of all  ${\sf \#sol}_{\phi,t}({\bf R},{\bf G}_{x},k')$, 
given the values of all  ${\sf \#sol}_{\phi,t}({\bf R},{\bf G}_{y},k')$, in  $O_{|\phi|+t}(k)$ steps.

\noindent{{\em Case 3:} $x$ has one child $y$ and  $\{v\}=\chi(y)\setminus \chi(x)$.} 
Given an ${\bf R}'\in {\cal R}_{\phi,t}$ where $|B({\bf R}')|=|B_{y}|$,
we define $v_{{\bf R}'}$ as the vertex of the underlying graph of ${\bf R}'$ 
that has the same index as the vertex $v$ in $G_{y}$.
For every such ${\bf R}\in{\cal R}_{\phi,t}$ 
we define 
$$\frak{P}_{x}({\bf R})=\{{\bf R}'\in {\cal R}_{\phi,t} \mid \mbox{$|B({\bf R}')|=|B_{y}|$~and~}(G({\bf R}'),B({\bf R}')\setminus \{v_{{\bf R}'}\},A({\bf R}'))\equiv_{\phi,r}{\bf R}\}.$$
Observe that
$$\#{\sf sol}_{{\phi},t}({\bf R},{\bf G}_{x},k')=\! \!\sum_{{\bf R}'\in\frak{P}_{x}({\bf R})}\!\!\#{\sf sol}_{{\phi},t}({\bf R}',{\bf G}_{y},k'),$$
therefore we can compute the values of all  ${\sf \#sol}_{\phi,t}({\bf R},{\bf G}_{x},k')$, 
given the values of all  ${\sf \#sol}_{\phi,t}({\bf R},{\bf G}_{y},k')$, in  $O_{|\phi|+t}(k)$ steps.

\noindent{{\em Case 4:} $x$ has two children $x_{1}$ and $x_{2}$}.  
For every ${\bf R}\in{\cal R}_{\phi,t}$ we define  $r({\bf R})=|B({\bf R})\cap A({\bf R})|$. We set
$$\frak{P}({\bf R})\!=\!\{({\bf R}_{1},{\bf R}_{2})\in {\cal R}_{\phi,t}\times{\cal R}_{\phi,t}\! \mid\! {\bf R}_{1}\sim{\bf R}_{2}~\wedge~ (G({\bf R}_{1}\oplus {\bf R}_{2}),B({\bf R}),A({\bf R}_{1}\oplus {\bf R}_{2}))\equiv_{\phi,r}{\bf R}\}$$
and observe that
$$\#{\sf sol}_{{\phi},t}({\bf R},{\bf G}_{x},k')=\! \!\sum_{(k_{1},k_{2})\in\Bbb{N}^{2}:\atop k_1+k_2=k'+r({\bf R})}\sum_{({\bf R}_{1},{\bf R}_{2})\in\frak{P}({\bf R})}\!\!\#{\sf sol}_{{\phi},t}({\bf R}_1,{\bf G}_{x_1},k_1) \cdot \#{\sf sol}_{{\phi},t}({\bf R}_2,{\bf G}_{x_2},k_2)$$
therefore we can compute the values of all  ${\sf \#sol}_{\phi,t}({\bf R},{\bf G}_{x},k')$, 
given the values of all  ${\sf \#sol}_{\phi,t}({\bf R},{\bf G}_{x_1},k')$
and  the values of all  ${\sf \#sol}_{\phi,t}({\bf R},{\bf G}_{x_2},k')$, in $O_{|\phi|+t}(k^2)$ steps.
\medskip

As the running time of the last case dominate the other two, we conclude 
that the above dynamic programming algorithm
can compute the value of ${\sf \#sol}_{\phi,t}({\bf R},{\bf G},k')$ for every $({\bf R},k')\in{\cal R}_{\phi,t}\times[0,k]$ in   $O_{|\phi|,t}(nk^2)$ steps.\end{proof}

We are now in position to prove Theorem~\ref{hobhouse}.

\begin{proof}[Proof of Theorem~\ref{hobhouse}]
We describe a polynomial size compactor $({\sf P},{\sf M})$ for $\Pi_{\phi,{\cal F}_{H}}.$ 
Given an input $(G,k)\in {\cal F}_{H}\times\Bbb{N}$, 
the condenser ${\sf P}$ of the compactor
runs as a first step the algorithm of Theorem~\ref{dishonor}.
If this algorithm 
reports that there is no set $A$ of size $k$ with $(G,k) \models \phi$,
the the condenser outputs ${\tt \$}$, i.e., ${\sf A}(G,k)={\tt \$}.$ Suppose now that the  output is  
a  $(tk,t,t)$-protrusion decomposition ${\bf G}_{1},\ldots, {\bf G}_{s}$ of $G$, along with the corresponding tree decompositions, for some constant $t$ that depends only on $h$ and $|\phi|.$
Let $K$ be the center of this protrusion decomposition and recall that $|K|,s\leq tk.$
We set $G_0=G[K]$ and let ${\bf G}_i=(G_i,B_i,-)$ for each $i\in [s].$
We also define ${\cal B}=\{B_{i},\mid i\in [s]\}$
where $B_{i}$ is the boundary of ${\bf G}_{i}$, $i\in[s].$ 
The next 
step of the condenser 
is to apply the algorithm of Lemma~\ref{supports}
and compute ${\sf \#sol}_{\phi,t}({\bf R},{\bf G}_i,k')$ for every $({\bf R},k',i)\in{\cal R}_{\phi,t}\times[0,k]\times[s]$,
in $O_{|\phi|+h}(nk^2)$ steps. 
The output of the condenser ${\sf P}$ is

$${\sf P}(G,k)=(G_0,{\cal B},\{{\sf \#sol}_{\phi,t}({\bf R},{\bf G}_i,k')\mid ({\bf R},k',i)\in{\cal R}_{\phi,t}\times[0,k]\times[s]\}).$$

\noindent Clearly, ${\sf P}(G,k)$ can be encoded in $O_{|\phi|+h}(k^{2})$ memory positions. 

We next describe the extractor {\sf M} of the compactor. 
 For simplicity, we write  $z:={\sf P}(G,k)$ and we define ${\sf M}({\tt \$})=0.$
We assume that there is a fixed labeling $\lambda$ of $G_0.$
The extractor ${\sf M}$ first computes the set ${\cal A}$ containing all subsets of $K$ of at most $k$ vertices. Notice that $|{\cal A}|=2^{O_{|\phi|+h}(k)}.$ Next, for each $A_0\in {\cal A}$, the algorithm builds 
the set ${\cal M}_{A_0}$ containing all mappings $\frak{m}: [s]\to {\cal R}_{\phi,t}$
with the property that, for every $i\in[s]$,
$(G_0,B_i,A_0)\sim \frak{m}(i).$
As the boundary of $\frak{m}(i)$ 
induces an identical labeled graph as $B_i$ does, we  denote $\frak{m}(i)$ as $(G_i^{\frak m}, B_i, A_i^{\frak m})$.
Notice that  $|{\cal M}_{A_0}|=2^{O_{|\phi|+h}(k)}$, for every $A_0\in{\cal A}.$

Let $A_0\in{\cal A}$ and $\frak{m}\in{\cal M}_{A_0}.$
For each such pair, the extractor runs a routine that 
constructs an annotated graph  $(D^{\frak{m}},A^{\frak{m}})$
as follows: first it initializes ${\bf D}^{\frak{m}}_{0}=(D_0,A^{\frak m}_0)$ with $D_0=G_0$ and $A^{\frak m}_0=A_0$. After constructing ${\bf D}^{\frak m}_i=(D_i,\bigcup_{j\in [i]}A^{\frak m}_j)$, 
the routine sets ${\bf D}^{m}_{i+1}= ( (D_i,B_{i+1},-)\oplus (G^{\frak m}_{i+1},B_{i+1},-), \bigcup_{j\in [i+1]}A^{\frak m}_j )$ iteratively from $i=0$ up to $s-1$. 
We set $(D^{\frak{m}},A^{\frak{m}})={\bf D}^{\frak{m}}_{s}.$ 
Notice that the routine runs in $O_{|\phi|+h}(k)$ steps and that $|D^{\frak{m}}|=O_{|\phi|+h}(k).$

 

The  extractor ${\sf M}$ is defined as
\begin{eqnarray*}
{\sf M}(z)\!=\!\!\!\!\sum_{A_0\in {\cal A}}\sum_{\frak{m}\in {\cal M}_{A_0}}\!\!\![(D^{\frak{m}},A^{\frak{m}})\models\phi]\cdot 
\Big(\!\!\sum_{\zeta\in {\cal K}_{k-|A_0|}}\prod_{i\in [s]}\#{\sf sol}_{{\phi},t}(\frak{m}(i),{\bf G}_{i},\zeta(i)+|B_{i}\cap A_0|) \Big)\!\!
\label{repealed}\end{eqnarray*}

\noindent where $[\cdot ]$ is a function indicating whether a sentence is true (=1) or false (=0), and 
${\cal K}_{\ell-|A_0|}$ is the set of all vectors $\zeta \in [0,k]^s$ such that $\sum_{i\in [s]} \zeta(i)=\ell-|A_0|.$

Having access to $\{{\sf \#sol}_{\phi,t}({\bf R},{\bf G}_i,k')\mid ({\bf R},k',i)\in{\cal R}_{\phi,t}\times[0,k]\times[s]\}$, we can 
compute ${\sf M}(z)$ in $2^{O_{|\phi|+h}(k)}$ steps.
Therefore, the extractor runs in the claimed running time.
It  remains to prove that ${\sf M}(z)$ equals  $|\{A\in{V(G)\choose k}\mid  (G,A)\models \phi\}|$. 

Before proceeding, we present a key claim.

\begin{claim}\label{egoistic}
Let ${\bf H}_i=(H_i,B,A_i)$ for $i=1,2$ be two compatible b-structures from ${\cal B}^{(t)}$. 
Let ${\bf H}'_2=(H'_2,B,A'_2)$ be a b-structure equivalent with ${\bf H}_2$. 
Then for every $B'\subseteq V(H_1)$ of size at most $t$, the two b-structures ${\bf D}$ and ${\bf D}'$ are equivalent under $\equiv_{\phi,t}$, where

\[
{\bf D}=((H_1,B,-)\oplus (H_2,B,-),B', A_1\cup A_2) \quad \text{and} \quad
{\bf D}'=((H_1,B,-)\oplus (H'_2,B,-),B', A_1\cup A'_2) 
\]
\end{claim}

\noindent{\em Proof of claim}: The compatibility between ${\bf D}$ and ${\bf D}'$ follows immediately from 
the fact $V(H_1)\cap (A_1\cup A_2)=V(H_1)\cap (A_1\cup A'_2)$, which is in turn implied by 
the equivalence (subsuming the compatibility) between ${\bf H}_2$ and ${\bf H}'_2$. 

Let ${\bf F}=(F,B',A)$ be an arbitrary b-structure of ${\cal B}^{(t)}$. Observe that the following annotated structures are identical:

\begin{align*}
{\bf F}\oplus {\bf D}&= (F,B',A)\oplus ((H_1,B,-)\oplus (H_2,B,-),B', A_1\cup A_2) \\
				&=((F,B',-)\oplus (H_1,B',-),B,A\cup A_1)\oplus (H_2,B,A_2)\\
				&=  ((F,B',-)\oplus (H_1,B',-),B,A\cup A_1)\oplus (H'_2,B,A'_2)=  {\bf F}\oplus {\bf D}'
\end{align*}

which implies ${\bf F}\oplus {\bf D}\models \phi$ if and only if ${\bf F}\oplus {\bf D}'\models \phi$.
\qed
\medskip

Let $A\subseteq V(G)$ be a set of size precisely $k$ with $(G,A)\models \phi.$ 
To see that $A$ contributes to ${\sf M}(z)$, we note that $A$ can be uniquely represented as 
the disjoint union $\bigcup_{i\in [0,s]} A_i$, where $A_0:=A\cap K$ and $A_i:=A\cap V(G_i)\setminus B_i$ for $i\in [s].$ 
Consider the mapping $\frak{m}:[s]\to {\cal R}_{\phi,t}$ defined as $\frak{m}(i)=\rep(G_i,B_i,A_i\cup (A_0\cap B_i) ).$ 
Clearly, $\rep(G_i,B_i,A_i\cup (A_0\cap B_i) )$ is compatible with $(G_0,B_i,A_0).$ Moreover, it is not difficult to 
see that Claim~\ref{egoistic} and the construction of $(D^{\frak{m}}, A^{\frak{m}})$ implies $[(D^{\frak m}, A^{\frak m})\models \phi]=[(G,A)\models \phi]$, and 
the vector $\zeta=(|A_1|,\ldots , |A_s|)$ is contained in ${\cal K}_{k-|A_0|}.$ 
Lastly, from $\frak{m}(i)=\rep(G_i,B_i,A_i\cup (A_0\cap B_i) )$ and $\zeta(i)+|B_{i}\cap A_0|=|A_i\cup (A_0\cap B_i)|$, 
the set $A_i\cup (A_0\cap B_i)$ contributes to $\#{\sf sol}_{{\phi},t}(\frak{m}(i),{\bf G}_{i},\zeta(i)+|A_0\cap B_{i}|)$ by 1 for each $i\in [s].$ 
Therefore, we know that $A$ contributes to the sum ${\sf M}(z)$ by 1. 
Furthermore, it is easy to see that distinct sets $A,A'\in {V(G) \choose k}$ with $(G,A),(G,A')\models \phi$ 
yield distinct contributions to ${\sf M}(z)$. 
That is, $|\{A\in{V(G)\choose k}\mid  (G,A)\models \phi\}|$ is at most the value of ${\sf M}(z)$. 

Conversely, consider a set $A_0\subseteq K$ of size at most $k$, a mapping $\frak{m}\in {\cal M}_{A_0}$ such that $[(D^{\frak{m}}, A^{\frak{m}})\models\phi ]=1$, 
and a vector $\zeta \in {\cal K}_{k-|A_0|}.$ 
We may assume $|A_0|\leq k$, since otherwise ${\cal K}_{k-|A_0|}=\emptyset.$ 
Note that for any $A'_i\subseteq V(G_i)$, the b-structure $(G_i,B_i,A'_i)$ is equivalent to $\frak{m}(i)$ under $\equiv_{\phi,t}$ only if 
they are compatible; this implies that  $B_i\cap A'_i = B_i\cap A_0.$ 
Therefore, any set $A'_i$ counted in $\#{\sf sol}_{{\phi},t}(\frak{m}(i),{\bf G}_{i},\zeta(i)+|B_{i}\cap A_0|)$ is of the form $(B_i\cap A_0) \dot{\cup} A_i$, where 
$A_i$ is a vertex subset of $V(G_i)\setminus B_i.$ Furthermore, such a set $A'_i$ satisfies $|A'_i|= \zeta(i)+|B_{i}\cap A_0|$ and thus we have $|A_i|=\zeta(i).$ 

Now, consider an arbitrary sequence $A'_1,\ldots , A'_s$ of vertex sets with $A'_i\subseteq V(G_i)$, each of which is counted in 
$\#{\sf sol}_{{\phi},t}(\frak{m}(i),{\bf G}_{i},\zeta(i)+|B_{i}\cap A_0|).$ Claim~\ref{egoistic}, $[(G^{\frak{m}},A^{\frak{m}})\models\phi ]=1$, and $\frak{m}_i\equiv_{\phi,t} (G_i,B_i, A')$ 
ensure that $(G,A_0\cup \bigcup_{i\in [s]}A'_i)\models \phi.$ Observe  that
$$|A_0\cup \bigcup_{i\in [s]}A'_i|= |A_0| + \sum_{i\in [s]} |A'_i\setminus B_i|=|A_0| + \sum_{i\in [s]} |A_i|=|A_0|+\sum_{i\in [s]} \zeta(i)=k.$$
That is, each combination of $A_0$, $\frak{m}$, $\zeta$, and a sequence $A'_1,\ldots , A'_s$ contributing 1 to the sum ${\sf M}(z)$, 
a vertex set $A$ of size precisely $k$ can be uniquely defined and we have $(G,A)\models \phi$. Clearly, distinct combinations lead to distinct such sets. 
Therefore, $|\{A\in{V(G)\choose k}\mid  (G,A)\models \phi\}|$ is at least the value of ${\sf M}(z)$. This completes the  proof. 
\end{proof}

\section{Conclusions}

Concerning Theorem~\ref{hobhouse},  we stress that the treewidth-modulability 
condition can be derived by other meta-algorithmic conditions. Such conditions are minor/contraction bidimensionality 
and linear separability for graphs excluding a graph/apex graph as a minor~\cite{FominLST10bidi,FominLST16bidi}. This  extends the 
applicability of our meta-algorithmic result to more problems but in more restricted graph classes. 
Natural follow-up questions are  whether the size of the compactor of can be made linear and whether its combinatorial applicability can be extended to more general graph classes.

We envision that the formal definition of a compactor that we give in this paper may 
encourage the research on data-reduction for counting problems. 
The apparent open issue is whether other problems (or families of problems)
may be amenable to this data-reduction paradigm (in particular, the results in~\cite{DiazST08effi,NishimuraRT05,Thurley07kern,ThurleyDipThes}
can be interpreted as results on polynomial compactors).

Another interesting  question is whether (and to which extent) 
the fundamental complexity results in~\cite{ChenFM11lowe,BodlaenderDFH09onpr,FortnowS11infe,Dell16andc,BodlaenderJK14kern,Drucker15newl,HarnikN10onth} on the non-existence of polynomial 
kernels  may have their counterpart for counting problems.

%

 \newcommand{\bibremark}[1]{\marginpar{\tiny\bf#1}}
  \newcommand{\biburl}[1]{\url{#1}}

\end{document}